\documentclass[a4paper,12pt]{article}

\usepackage{complexity} 
\usepackage{enumerate}
\usepackage{amsmath, amssymb, amsthm,complexity}
\usepackage{authblk}
\usepackage{todonotes}

\usepackage{vmargin}
\setmarginsrb{1in}{1in}{1in}{1in}{0mm}{0mm}{0mm}{7mm}

\usepackage{mathtools}
\usepackage{thmtools}
\usepackage{thm-restate}
\usepackage{hyperref}
\usepackage{cleveref}
\usepackage{complexity}
\usepackage{authblk}
\theoremstyle{plain}
\usepackage{microtype}
\usepackage{multirow}
\usepackage{amsmath,amssymb}
\usepackage{comment}
\usepackage{enumerate}
\usepackage{xspace}
\usepackage{listings}
\usepackage{color}
\usepackage{cite}
\usepackage{graphicx}
\RequirePackage{fancyhdr}
 \usepackage{xcolor}
\usepackage{boxedminipage}

\newcommand{\mypara}[1]{\smallskip\noindent{\textbf{\sffamily #1} \ }}

\newcommand{\tw}{{\normalfont \textsf{tw}}}

\newcommand{\cw}{{\normalfont \textsf{cw}}}

\newcommand{\hG}{{\hat G}}
\newcommand{\hu}{{\hat u}}

\newcommand{\cvctreewidth}{{\sc CVC($\eta$-transversal)}}
\newcommand{\cvcchordal}{{\sc CVC({\sc  Chordal-Del})}}
\newcommand{\cvcGdeletion}{{\sc CVC(${\cG}$-Deletion)}}

\newcommand{\cO}{{\cal O}}

\newcommand{\cB}{{\mathcal{B}}}

\newcommand{\cD}{\mathcal{D}}

\newcommand{\cH}{\mathcal{H}}
\newcommand{\cG}{\mathcal{G}}

\newcommand{\cT}{\mathcal{T}}

\newcommand{\OPT}{\mbox{\rm OPT}}

\newcommand{\eps}{\varepsilon}

\newcommand{\CC}{{\mathcal C}}
\newcommand{\GG}{{\mathcal G}}

\newcommand{\shortversion}[1]{}

\newcommand{\CVC}{{\sc Connected Vertex Cover}}

\newcommand{\OO}{{\mathcal O}}
\newcommand{\bigoh}{{\mathcal O}}

\newenvironment{tightcenter}
 {\parskip=0pt\par\nopagebreak\centering}
 {\par\noindent\ignorespacesafterend}

\usepackage{tikz}
\usetikzlibrary{calc}
\usepackage{xargs}
\usepackage{xifthen}
\usepackage{framed}

\usepackage{ctable}

\newlength{\RoundedBoxWidth}
\newsavebox{\GrayRoundedBox}
\newenvironment{GrayBox}[1]%
   {\setlength{\RoundedBoxWidth}{\textwidth-4.5ex}
    \def\boxheading{#1}
    \begin{lrbox}{\GrayRoundedBox}
       \begin{minipage}{\RoundedBoxWidth}%
   }{%
       \end{minipage}
    \end{lrbox}%
    \begin{tightcenter}%
    \begin{tikzpicture}%
       \node(Text)[draw=black!60,fill=white,rounded corners,%
             inner sep=2ex,text width=\RoundedBoxWidth]%
             {\usebox{\GrayRoundedBox}};
        \coordinate(x) at (current bounding box.north west);
        \node [draw=white,rectangle,inner sep=3pt,anchor=north west,fill=white] 
        at ($(x)+(6pt,.75em)$) {\boxheading};
    \end{tikzpicture}
    \end{tightcenter}\vspace{0pt}%
    \ignorespacesafterend
}    

\newenvironment{problem}[2][]{\noindent\ignorespaces%
                                \FrameSep=6pt%
                                \parindent=0pt%
                \vspace*{-.5em}
                \ifthenelse{\isempty{#1}}{%
                  \begin{GrayBox}{\textsc{#2}}%
                }{%
                }
                \newcommand\Prob{{\sf Problem:}}%
                \newcommand\Input{{\sf Input:}}%
                \newcommand\Parameter{{\sf Parameter:} }          
                \begin{tabular*}{\textwidth}{@{\hspace{.1em}} >{\itshape} p{1.2cm} p{0.85\textwidth} @{}}%
            }{
                \end{tabular*}%
                \end{GrayBox}%
                \vspace*{-.5em}
                \ignorespacesafterend
            } 

\DeclarePairedDelimiter{\ceil}{\lceil}{\rceil}

\newtheorem{theorem}{\bf Theorem}
\newtheorem{definition}{\bf Definition}

\newtheorem{observation}{\bf Observation}
\newtheorem{reduction rule}{\bf Reduction Rule}
\newtheorem{lemma}{\bf Lemma}

\title{On the Approximate Compressibility of Connected Vertex Cover\footnote{A part of this work was done when the first author was affiliated to The Institute of Mathematical Sciences, HBNI, Chennai, India.}}

\date{}
\author[1]{Diptapriyo Majumdar\footnote{Corresponding author.}}
\author[2]{M. S. Ramanujan}
\author[3]{Saket Saurabh}
\affil[1]{Royal Holloway, University of London, United Kingdom
    \texttt{diptapriyo.majumdar@rhul.ac.uk}}
\affil[2]{University of Warwick, United Kingdom
\texttt{R.Maadapuzhi-Sridharan@warwick.ac.uk}}
\affil[3]{The Institute of Mathematical Sciences, HBNI, Chennai, India
  \texttt{saket@imsc.res.in}}

\begin{document}

\maketitle

\begin{abstract}
The {\CVC} problem, where the goal is to compute a minimum set of vertices in a given graph which forms a vertex cover and induces a connected subgraph, is a fundamental combinatorial problem and has received extensive attention in various subdomains of algorithmics. In the area of kernelization, it is known that this problem is unlikely to have a polynomial kernelization algorithm. However, it has been shown in a recent work of Lokshtanov et al.~[STOC 2017] that if one considered an appropriate notion of {\em approximate kernelization}, then this problem parameterized by the solution size does admit an {\em approximate} polynomial kernelization. In fact, Lokshtanov et al.~were able to obtain a polynomial size approximate kernelization scheme (PSAKS) for {\CVC} parameterized by the solution size. A PSAKS is essentially a preprocessing algorithm whose error can be made arbitrarily close to 0.

In this paper we revisit this problem, and consider parameters that are {\em strictly smaller} than the size of the solution and obtain the first polynomial size approximate kernelization schemes for the {\CVC} problem when parameterized by the deletion distance of the input graph to the class of cographs, the class of bounded treewidth graphs, and the class of all chordal graphs.
\end{abstract}

\section{Introduction}
\label{sec:intro}

Polynomial-time preprocessing is one of the widely used methods to tackle NP-hardness in practice, and {\em kernelization} has been extremely successful in laying down a mathematical framework for the design and rigorous analysis of preprocessing algorithms for decision problems. 
The central notion in kernelization is that of a \emph{kernelization}, which is a preprocessing algorithm that takes as input a parameterized problem, which is a pair $(I,k)$, where $I$ is the problem instance and $k$ is an integer called the {\em parameter}. A kernelization is required to run in polynomial-time and convert a potentially large input $(I,k)$ into an equivalent instance $(I',k')$ such that $|I'|$ and $k'$ are both bounded by a function of the parameter $k$. Over the last decade, the area of kernelization has seen the development of a wide range of tools to design preprocessing algorithms and a rich theory of lower bounds has been developed based on assumptions from complexity theory~\cite{BodlaenderDFH09,DellM10,BJK11,FortnowS11,DellM12,Drucker12,HermelinW12,DomLS14,HermelinKSWW15,kernel-book}. We refer the reader to the survey articles by Kratsch~\cite{Kratsch14} or Lokshtanov et al.~\cite{LMS2012} for relatively recent developments, or the textbooks~\cite{CyganFKLMPPS15,DF2012Book}, for an introduction to the field. 

An `efficient preprocessing algorithm' in this setting is referred to as a \emph{polynomial kernelization} and is simply a kernelization whose output has size polynomially upper bounded in the value of the parameter of the input. The central classification task in the area is to classify each NP-hard problem as one which has a polynomial kernel, or as one that does not. 

The {\sc Vertex Cover} problem is one of the most frequently studied problems from the point of view of kernelization and buoyed by the rich literature on kernelization for {\sc Vertex Cover} parameterized by the solution size, researchers have more recently turned their attention to the design of kernelization algorithms for {\sc Vertex Cover} parameterized by smaller parameters. The results most relevant to us in this line of enquiry are the polynomial kernelization given by Jansen and Bodlaender~\cite{JansenB13} for {\sc Vertex Cover} parameterized by the size of the feedback vertex set and the result of Cygan et al.~\cite{CyganLPPS14}, in which they showed that {\sc Vertex cover} is unlikely to have a polynomial kernelization when parameterized by the deletion distance of the given graph to the class of graphs of treewidth at most $\eta$, for any $\eta>1$. Here, the deletion distance of the given graph $G$ to any graph class $\cG$ is the size of the smallest set $S\subseteq V(G)$ such that $G-S\in \cG$.

On the other hand, the {\CVC} problem, where the solution is also required to induce a connected subgraph of the input graph, is known to exclude a polynomial kernelization already when parameterized by the solution size, under standard complexity-theoretic hypotheses~\cite{DomLS14}. However, the study of preprocessing for this problem was handed a new lease of life by the recent work of Lokshtanov et al.~\cite{LokshtanovPRS17}, who aimed to facilitate the rigorous analysis of preprocessing algorithms in conjunction with approximation algorithms via the introduction of {\em $\alpha$-approximate kernels}. 

\begin{quote}
Informally speaking, an $\alpha$-approximate kernel is a polynomial-time algorithm that, given an instance $(I,k)$ of a parameterized problem, outputs an instance $(I',k')$ such that $|I'|+ k' \leq g(k)$ for some computable function $g$ and any $c$-approximate solution to the instance $(I',k')$ can be turned in polynomial-time into a $(c \cdot \alpha)$-approximate solution to the original instance $(I,k)$. 
\end{quote}

As earlier, the notion of `efficiency' in this context is captured by the function $g$ being polynomially upper bounded, in which case we call this algorithm an {\em $\alpha$-approximate polynomial kernelization}. 
We refer the reader to Section~\ref{sec:prelims} for a formal definition of all terms related to (approximate) kernelization. 

In their work, Lokshtanov et al.~\cite{LokshtanovPRS17} considered several problems which are known to exclude polynomial kernels and presented an $\alpha$-approximate polynomial kernel for these problems for {\em every} fixed $\alpha>1$, also called a polynomial size approximate kernelization scheme (PSAKS, see Section~\ref{sec:prelims} for formal definition). This implies that allowing for an arbitrarily small amount of error while preprocessing can drastically improve the extent to which the input instance can be reduced, even when dealing with problems for which polynomial kernels have been ruled out under the existing theory of lower bounds.

\begin{figure}[t]
\centering
    \includegraphics[scale=0.30]{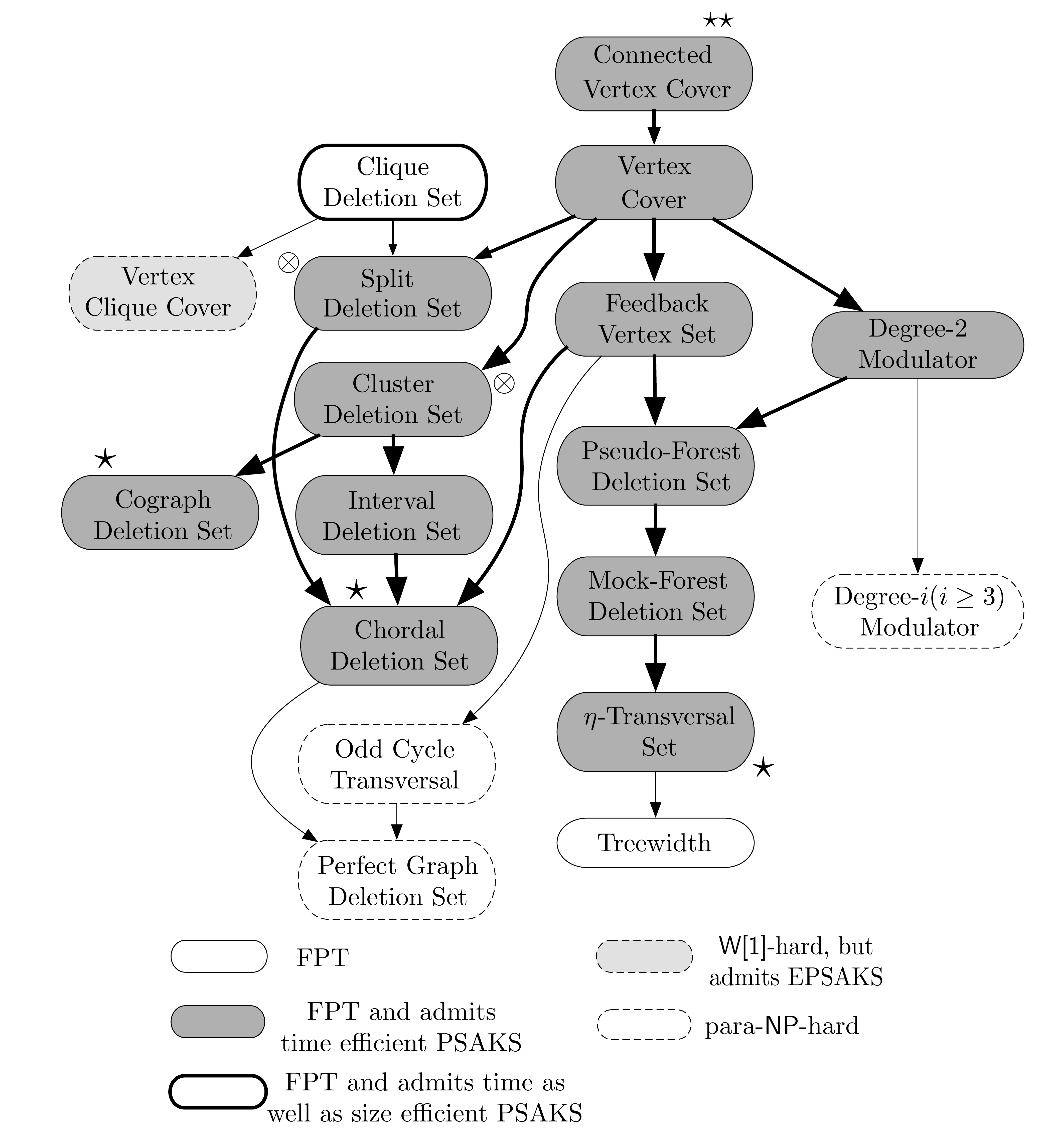}
    \caption{Hierarchy of Parameters. An arrow from parameter $x$ to parameter $y$ means $y \leq f(x)$ for some polynomial $f$. Results marked $\star$ indicates the ones considered in this paper. Results marked $\star\star$ appeared in~\cite{LokshtanovPRS17}, and marked $\otimes$ appeared in~\cite{KMR18}. See Section~\ref{sec:prelims} and~\cite{LokshtanovPRS17} for definitions.}
    \label{fig:parameter-hierarchy}
\end{figure} 

In particular, they showed that the {\sc Connected Vertex Cover} problem admits an $\alpha$-approximate polynomial kernel for every $\alpha>1$. We believe that their result provides a promising starting point towards a comprehensive study on approximate kernelizations for {\CVC}, with the aim being the replication of the success enjoyed by {\sc Vertex Cover} in the domain of kernelization. Consequently, we consider the question of designing $\alpha$-approximate kernelizations for {\CVC} in a systematic manner by considering as our parameter, the deletion distance of the given graph to well understood super classes of edgeless graphs. We point out that we are not the first to attempt this. Krithika et al.~\cite{KMR18} obtained a PSAKS for this problem when parameterized by the deletion distance of the input graph to the class of split graphs.
Our results generalize their result and also provide unified approximate kernelizations for this problem with respect to several parameters, including the deletion distance of the input graph to the class of split graphs and cographs.

\subsection{Our results and significance of the chosen parameterizations }
The parameters we consider in this paper are the deletion distances of the input graph to 
(a) bounded treewidth graphs, (b) split graphs or cographs, and (c) chordal graphs.
Since an edgeless graph is contained in all of these graph classes, it follows that all of our parameters are {\em upper bounded} by the minimum vertex cover (note the removal of the connectivity requirement) in any given graph. Clearly, the size of the smallest vertex cover is in turn upper bounded by the size of the smallest {\em connected} vertex cover. Consequently, all our parameters are upper bounded by the standard parameter for {\CVC}, the solution size.
Moreover, since the classes of bounded treewidth graphs, cographs, and chordal graphs are all pairwise incomparable, it follows that our parameters are also pairwise incomparable. See Figure~\ref{fig:parameter-hierarchy} for a hierarchy of the parameters.

\mypara{Parameterization by deletion distance to bounded treewidth graphs.}
One of our first two results concerns PSAKS for {\CVC} parameterized by the deletion distance to graphs with treewidth $\eta$ (Section~\ref{sec:prelims} contains the formal definitions of a PSAKS).
In our first result, we demonstrate the existence of something stronger -- a \emph{time efficient} PSAKS for a more general problem which we call {\cvcGdeletion}, defined as below. Here, $\cG$ is a {\em fixed hereditary graph class}.

\begin{problem}[]{{\sc CVC parameterized by ${\cG}$-Deletion ({\cvcGdeletion})}}
  \Input & A graph $G$, a vertex set $S$ of size $k$ such that $G - S \in {\cG}$, integer $p$. \\
  \Parameter & \hspace{5 pt} $k$ \\
  \Prob & Does $G$ have a connected vertex cover of size at most $p$?
\end{problem}

\begin{restatable}{theorem}{CVCGeneralTheorem}
\label{thm:cvc-G-deletion-theorem}
Suppose that a graph class ${\cG}$ is polynomial-time recognizable.
For every $0 < \eps < 1$, 
{\cvcGdeletion} admits a {\sc PSAKS} with $k^{\OO(1/\eps)}$ vertices if {\CVC} has a {\sc PTAS} on the graph class ${\GG} + 1\cdot {\sf v}$. Moreover, if this {\sc PTAS} is an Efficient {\sc PTAS}, then the {\sc PSAKS} is a time efficient {\sc PSAKS}.
\end{restatable}

The class ${\cG} + 1\cdot {\sf v}$ is simply the class of all graphs from which a single vertex can be removed to obtain a graph in ${\cG}$.
We refer the reader to \cite{Vazirani-Approx} for the definitions of PTAS and Efficient PTAS. 
Now, as a consequence of Theorem~\ref{thm:cvc-G-deletion-theorem} and the fact that {\CVC} has a linear time algorithm on graphs of constant treewidth, we have our second result.
 
\begin{restatable}{cor}{TreewidthResult}
\label{cor:cvc-treewidth-result}
For every fixed $\eta$, 
{\cvctreewidth} has a time efficient {\sc PSAKS} with $k^{\OO(1/\eps)}$ vertices.
\end{restatable}

This result provides an interesting contrast to the result of Cygan et al.~\cite{CyganLPPS14} which rules out polynomial kernelizations even for {\sc Vertex Cover} (with no connectivity requirements) parameterized by the deletion distance to treewidth $\eta$ graphs (for $\eta>1$).

\mypara{Parameterization by deletion distance to chordal graphs.} For our third result, we consider the parameterization by the deletion distance of the input graph to another class of graphs which is incomparable with both bounded treewidth and bounded diameter graphs. This is the class of chordal graphs. The central idea driving this result is a new reduction rule for {\CVC}, the exhaustive application of which leaves us with an equivalent instance of {\cvctreewidth} for an appropriate $\eta$ depending only on $\eps$. By combining this reduction rule with Corollary~\ref{cor:cvc-treewidth-result}, we obtain the following result.

\begin{restatable}{theorem}{chordalmaintheorem}\label{thm:chordal_main}
{\cvcchordal} has a time efficient {\sc PSAKS} with $k^{\OO(1/\eps)}$ vertices.
\end{restatable}

\mypara{Parameterization by deletion distance to split graphs or cographs.}
After our result on {\cvcchordal}, we obtain a PSAKS for {\CVC} parameterized by the deletion distance of the input graph to the class of split graphs and cographs.
More specifically, we consider the parameter $|S|$ such that every connected component of $G - S$ is either a split graph or a cograph.

\begin{restatable}{cor}{diametermaintheorem}\label{thm:diameter_main}
{\CVC} admits a time efficient PSAKS when parameterized by deletion distance to a graph whose connected components are either split graph or cograph.
\end{restatable}

While the result in the case of split graphs provides an alternate proof to the one given by Krithika et al.~\cite{KMR18}, our PSAKS for {\CVC} parameterized by the deletion distance to cographs is the first such result. 
Finally, we prove that our three main results can in fact be unified under a single even stronger parameterization which is the deletion distance of the input graph to the class of graphs where every connected component is \emph{either} a treewidth-$\eta$ graph {\em or} a chordal graph {\em or} a cograph. 

\begin{restatable}{theorem}{FinalResultMain}
\label{thm:more-general-theorem-statement}
For every fixed $\eta\in {\mathbb N}$, {\CVC} parameterized by the size $k$ of a modulator to the class of graphs where each connected component is a cograph or a chordal graph or a graph of treewidth at most $\eta$ admits a time efficient {\sc PSAKS} with $k^{\OO(1/\eps)}$ vertices.
\end{restatable}

When this deletion set is part of the input we can directly utilize our framework to also encapsulate graph classes which are significantly more general than classes comprising graphs which have a small deletion distance to {\em only one} of $\{\tw\mbox{-}\eta,{\sf Chordal},{\sf Cograph}\}$.
For instance, one can easily observe the existence of graphs from which an unbounded number of vertices must be deleted in order to move them into {\em any one} of $\{\tw\mbox{-}\eta,{\sf Chordal},{\sf Cograph}\}$ but at the same time only a constant number of vertices need to be deleted from them in order to obtain a graph where each connected component lies in one of $\{\tw\mbox{-}\eta,{\sf Chordal},{\sf Cograph}\}$. We refer the reader to Section~\ref{sec:conclusion} for a more detailed discussion. 

\mypara{Related work on kernelizations for {\sc Vertex Cover} with respect to parameters smaller than solution size.}~Kratsch and Wahlstr\"{o}m~\cite{KratschW12} gave the first (randomized) polynomial kernelization for {\sc Vertex Cover} parameterized above the optimum value of the standard LP relaxation. This result was later strengthened by Kratsch~\cite{Kratsch16} who parameterized above an even stronger lower bound on the solution, $2{\sf LP}-{\sf MM}$, where ${\sf LP}$ denotes the optimum value of the standard LP relaxation and ${\sf MM}$ denotes the size of the maximum matching in the input graph. Majumdar et al.~\cite{MajumdarRS15} considered as their parameter the deletion distance to graphs of degree 2 and to cluster graphs where each clique has bounded size. 
Fomin and Str{\o}mme~\cite{FominS16} considered the deletion distance to pseudo-forests, which are graphs where every connected component has at most one cycle. They also 
showed that parameterizing {\sc Vertex Cover} by the deletion distance even to cactus graphs is unlikely to lead to a polynomial kernel. More recently, Kratsch and Hols~\cite{KratschHols17} generalized the positive results of~\cite{MajumdarRS15,FominS16} by choosing as their parameter the deletion distance of the input graph to $d$-quasi-forests, which are graphs where every connected component has a set of at most $d$ vertices whose deletion leaves a forest. Finally, Bougeret and Sau~\cite{BougeretSau17} focussed on the deletion distance to graphs of treedepth $c$, for any fixed constant $c$, and obtained a polynomial kernel for {\sc Vertex Cover}.

\section{Preliminaries}
\label{sec:prelims}

We use $[r]$ to denote the set $\{1,2,\ldots,r\}$.

\mypara{Graph theoretic preliminaries.}
For $\ell\in {\mathbb N}$, we denote by $P_\ell$ the path on $\ell$ vertices.
All graphs studied in this paper are undirected. 
A graph is a {\it cluster graph} if each of its connected components is a clique. A graph is a {\it co-cluster graph} if its complement is a cluster graph.
A graph is a {\it split graph} if its vertex set can be partitioned into two sets one inducing a clique and the other inducing an independent set.
A graph is called a {\it cograph} if it has no induced $P_4$.
A graph is called {\em chordal} if it has no induced cycle of length more than $3$.
Given an undirected connected graph, for every pair of vertices $u,v \in V(G)$, we use $dist(u,v)$ to denote the length of a shortest path from $u$ to $v$.
The {\em diameter} of $G$ is $\max_{\{u,v\} \in {{V(G)}\choose{2}}} dist(u,v)$.
In other words, the diameter of a graph is the largest among all shortest distances between any pair of vertices.
Given a graph $G$ and vertex set $X$, we define the operation of {\em identifying} the set $X$ as the construction of the following graph $G'$.
The vertex set of $G'$ is $V(G')=V(G)\setminus X\cup \{\hat x\}$, where $\hat x$ is a new vertex not in $G$.
The edges of $G'$ are defined as follows. For every $u,v\in V(G)\setminus X$, if $uv\in E(G)$, then $uv\in E(G')$ as well. For every $u\in V(G)\setminus X$ and $v\in X$, if $uv \in E(G)$, then $u\hat{x}\in E(G')$.
We ignore all edges of $G$ with both endpoints in $X$.
We introduce and define the following notation that we will use throughout the rest of the paper.
We use $\cw (G)$ to denote the {\em cliquewidth} of the graph $G$ (see~\cite{Courcelle1990} for the definition of cliquewidth).

\begin{definition}
\label{defn:graph-class-modulator}
Let $\cG$ be a graph class. For $r\in {\mathbb N}$, we denote by $\cG +r \cdot {\sf v}$ the class of graphs in which there is a set of at most $r$ vertices whose deletion results in a graph from $\cG$. We denote by $\CC\CC(\cG)$ the class of all graphs whose connected components lie in $\cG$ and by $\CC\CC(\cG) +r \cdot {\sf v}$ the class of graphs in which there is a set of at most $r$ vertices whose deletion results in a graph where every connected component is in $\cG$.
\end{definition}

Let $S \subseteq V(G)$ be a set of vertices such that $G - S \in {\cG}$.
Then we say that $S$ is a {\em ${\cG}$-deletion set} of $G$.

\mypara{Tree Decomposition and Treewidth.} Let $G$ be a graph. A {\em tree decomposition} of $G$ is a pair $(T,\mathcal{X}=\{X_{t}\}_{t\in V(T)})$ where $T$ is a tree and ${\cal X}$ is a collection of subsets of $V(G)$ such that (a)
$\forall {e=uv \in E(G)}, \exists {t\in V(T)} : \{u,v\} \subseteq X_{t}$ and 
(b) $\forall {v\in V(G)}, \ T[\{t\mid v\in X_{t}\}]$ is a non-empty connected subtree of $T$. 
The {\em width} of $(T,{\cal X})$ is 
 defined as $\max_{}\{|X_t|-1\mid {t\in V(T)}\}$ and the {\em treewidth} of $G$ is the minimum width over all tree decompositions of $G$ and is denoted by $\tw(G)$. 
An {\em $\eta$-transversal} of a graph $G$ is a subset $X$ such that $\tw(G-X)\leq \eta$.
To perform dynamic programming over tree decomposition, we consider {\em nice tree decomposition} that is a tree decomposition satisfying some specific properties.
Due to Cygan et al.~\cite{CyganFKLMPPS15}, given a tree decomposition of width $\eta$, there exists a polynomial-time algorithm that converts the given tree decomposition into a nice tree decomposition of width $\eta$ in polynomial-time.

\begin{definition}{\rm [\sf Nice Tree Decomposition]}
\label{defn:nice-tree-decomposition}
Let $G$ be a graph.
A tree decomposition $\cT = (T,\mathcal{X}=\{X_{t}\}_{t\in V(T)})$ is said to be a {\em nice tree decomposition} if it is rooted, every node has at most two children, and each of its nodes are of one of the following types.
\begin{itemize}
	\item {\bf Root node:} If $r$ is the {\em root node} of $\cT$, then $X_r = \emptyset$.
	Informally, the bags corresponding to the root node is an empty bag.
	\item {\bf Leaf node:} If $t \in V(T)$ is a {\em leaf node} of $\cT$, then $X_t = \emptyset$.
	Informally, the bag corresponding to leaf nodes are also empty bags.
	\item {\bf Introduce node:} Suppose that $t_2$ is the only child of $t_1$ in $\cT$.
	Then, $t_1$ is called an {\em introduce node} when there exists $u \notin X_{t_2}$ such that $X_{t_1} = X_{t_2} \cup \{u\}$.
	Informally speaking, an introduce node adds a new vertex.
	\item {\bf Forget node:} Suppose that $t_2$ is the only child of $t_1$ in $\cT$.
	Then, $t_1$ is called a {\em forget node} when there exists $u \in X_{t_2}$ such that $X_{t_1} = X_{t_2} \setminus \{u\}$.
	Informally speaking, a forget node removes a vertex.
	\item {\bf Join node:} Let $t$ has two children $t_1$ and $t_2$.
	Then, $t$ is called a {\em join node} if $X_t = X_{t_1} = X_{t_2}$.
	Informally, a join node joins two bags with equal contents.
\end{itemize}
\end{definition}

We also provide the definition of {\em clique tree decomposition} that we use later in our paper.

\begin{definition}{\rm[\sf Clique Tree Decomposition]}
\label{defn:clique-tree-decomposition}
A tree decomposition $\cT = (T,\mathcal{X}=\{X_{t}\}_{t\in V(T)})$ of a graph $G$ is called a {\em clique tree decomposition} of $G$ if it is a valid tree decomposition of $G$, and for all $t \in V(T)$, $G[X_t]$ induces a clique in $G$.
Informally, a clique tree decomposition of a graph is a valid tree decomposition every bag of which induces a clique.
\end{definition}
 
\begin{theorem}[~\cite{Golumbic}]
\label{lemma:chordal-clique-tree}
The number of maximal cliques of a chordal graph is upper bounded by a polynomial in the size of the graph, and can be enumerated in polynomial-time.
These maximal cliques form a clique tree decomposition of a chordal graph.
Furthermore, this clique tree decomposition of a chordal graph can also be obtained in polynomial-time.
\end{theorem}

\mypara{Parameterized algorithms and kernels.}
A parameterized problem $\Pi$ is a subset of $\Gamma^{*}\times \mathbb{N}$ for some finite alphabet $\Gamma$. An instance of a parameterized problem is a pair $(x,k)$, where $k$ is called the parameter and $x$ is the input. We assume that $k$ is {\em given in unary} and without loss of generality, $k\leq |x|$. 
The notion of {\em kernelization} is formally defined as follows. 

\begin{definition}{\rm [\sf Kernelization]} 
Let $\Pi\subseteq \Gamma^{*}\times \mathbb{N}$ be a parameterized problem and $g$ be a computable function.
We say that $\Pi$ {\em admits a kernel of size $g$} if there exists an algorithm, referred to as a {\em kernelization} (or a {\em kernel}) that, given $(x,k)\in \Gamma^{*}\times \mathbb{N}$, outputs in time polynomial in $|x|+k$, a pair $(x',k')\in \Gamma^{*}\times \mathbb{N}$ such that
\begin{itemize}
\item[(a)] $(x,k)\in \Pi$ 
if and only if $(x',k')\in \Pi$, and 
\item[(b)] $\max\{|x'|, k' \}\leq g(k)$.
\end{itemize}
When $g(k)=k^{\cO(1)}$, we say that $\Pi$ {\em admits a polynomial kernel}.
\end{definition}

\mypara{Parameterized optimization problems and approximate kernels.}

\begin{definition}{\rm \cite{LokshtanovPRS17}}\label{def:paraOptProblem}
A parameterized optimization (minimization or maximization) problem is a computable function
$\mathrm{\Pi}~:~\Sigma^*\times \mathbb{N}\times \Sigma^*\rightarrow {\mathbb R}\cup \{\pm\infty\}.$
\end{definition}
The {\em instances} of a parameterized optimization problem $\mathrm{\Pi}$ are pairs $(I,k) \in \Sigma^*\times \mathbb{N}$, and a {\em solution} to $(I,k)$ is simply a string $s \in \Sigma^*$, such that $|s| \leq |I|+k$. The {\em value} of the solution $s$ is $\mathrm{\Pi}(I,k,s)$. 

Since the problems we deal with in this paper are all minimization problems, we state some of the definitions only in terms of minimization problems when the definition for maximization problems is analogous. 
 As an illustrative example, we provide the definition of the parameterized optimization version of {\sc Connected Vertex Cover} parameterized by the solution size (see~\cite{LokshtanovPRS17}). This is a minimization problem, where the optimization function ${\sc CVC}:\Sigma^*\times \mathbb{N}\times \Sigma^*\rightarrow {\mathbb R}\cup \{\infty\}$ is defined as follows.

\[{\sc CVC}(G,k,S) = \left\{
\begin{array}{rl}
\infty & \text{if $S$ is not a connected vertex cover of $G$},\\
\text{min}\{|S|, k + 1\} & \text{otherwise.}
\end{array}
\right.
\]

\begin{definition}{\rm \cite{LokshtanovPRS17}}\label{def:paraOpt}
For a parameterized minimization problem $\mathrm{\Pi}$, the {\em optimum value} of an instance $(I,k) \in \Sigma^*\times \mathbb{N}$ is
$OPT_\mathrm{\Pi}(I,k) = \min_{\substack{s \in \Sigma^* \\ |s| \leq |I|+k}} \mathrm{\Pi}(I,k,s)$.
\end{definition}

Consequently, in the case of {\sc Connected Vertex Cover} above, we define $\OPT(G,k)=\min_{S\subseteq V(G)}{\sc CVC}(G,k,S)$. 
We now recall other relevant definitions from \cite{LokshtanovPRS17} regarding {\em approximate kernels}.

\begin{definition}{\rm \cite{LokshtanovPRS17}}\label{def:polyTimePreProcessAppx}
Let $\alpha \geq 1$ be a real number and $\mathrm{\Pi}$ be a parameterized minimization problem. An {\em $\alpha$-approximate-polynomial-time preprocessing algorithm} ${\cal A}$ for $\mathrm{\Pi}$ is a pair of polynomial-time algorithms. The first one is called the {\em reduction algorithm}, and computes a map ${\cal R}_{\cal A} : \Sigma^* \times \mathbb{N} \rightarrow \Sigma^* \times \mathbb{N}$. Given as input an instance $(I,k)$ of $\mathrm{\Pi}$ the reduction algorithm outputs another instance $(I',k') = {\cal R}_{\cal A}(I,k)$.

The second algorithm is called the {\em solution lifting algorithm}. This algorithm
takes as input an instance $(I,k) \in \Sigma^* \times \mathbb{N}$ of $\mathrm{\Pi}$, the output instance $(I',k')$ of the reduction algorithm, and a solution $s'$ to the instance $(I',k')$. The solution lifting algorithm works in time polynomial in $|I|$,$k$,$|I'|$,$k'$ and $|s'|$, and outputs a solution $s$ to $(I,k)$ such that the following holds.

$$
\frac{\mathrm{\Pi}(I,k,s)}{OPT(I,k)} \leq \alpha \cdot \frac{\mathrm{\Pi}(I',k',s')}{OPT(I',k')}.
$$

The {\em size} of a polynomial-time preprocessing algorithm ${\cal A}$ is a function $\text{size}_{\cal A} : \mathbb{N} \rightarrow \mathbb{N}$ defined as $\text{size}_{\cal A}(k) = \sup\{|I'| + k' : (I',k') = {\cal R}_{\cal A}(I,k), I \in \Sigma^*\}$.
\end{definition}

\begin{definition}{\rm \cite{LokshtanovPRS17}}\label{def:appxkernel}~{\rm [\sf Approximate Kernelization]} 
An {\em $\alpha$-approximate kernelization} (or {\em $\alpha$-approximate kernel}) for a parameterized optimization problem $\mathrm{\Pi}$, and real $\alpha \geq 1$, is an $\alpha$-approximate-polynomial-time preprocessing algorithm ${\cal A}$ for $\mathrm{\Pi}$ such that $\text{size}_{\cal A}$ is upper bounded by a computable function $g : \mathbb{N} \rightarrow \mathbb{N}$. We say that ${\cal A}$ is an $\alpha$-approximate polynomial kernelization if $g$ is a polynomial function. 
\end{definition}

\begin{definition}{\rm \cite{LokshtanovPRS17}}\label{def:psaks}~{\rm [\sf Approximate Kernelization Schemes]} 
A {\em polynomial size approximate kernelization scheme} (PSAKS) for a parameterized optimization problem $\mathrm{\Pi}$ is a family of ${\alpha}$-approximate polynomial kernelization algorithms, with one such algorithm for every $\alpha > 1$.
\end{definition}

Now, an approximate kernelization scheme can be of different types.
We describe them as follows.

\begin{definition}{\rm [\sf Time Efficient PSAKS]} 
A PSAKS is said to be {\em time efficient} when both the followings hold.
\begin{itemize}
	\item The running time of the reduction algorithm is upper bounded by $f(\alpha)\cdot |I|^c$ for a function $f$ on $\alpha$ and some constant $c$ independent of $|I|,k$ and $\alpha$.
	\item The running time of the solution lifting algorithm is upper bounded by $g(\alpha)\cdot |I|^c$ for a function $g$ on $\alpha$ and some constant $c$ independent of $|I|,k$ and $\alpha$.
\end{itemize}
\end{definition}

\begin{definition}{\rm [\sf Size Efficient PSAKS]}
A {\em size efficient PSAKS} or simply an {\em efficient PSAKS (EPSAKS)} is a PSAKS such that the size of the output instance when run on $(I,k)$ with approximation parameter $\alpha$ can be upper bounded by $f(\alpha)k^{c}$ for a function $f$ on $\alpha$ and a constant $c$ independent of $|I|, k$ and $\alpha$.
\end{definition}

For further details on approximate kernelizations, we refer the reader to~\cite{LokshtanovPRS17}.

\section{{\CVC} parameterized by ${\cG}$-Deletion Number}
\label{sec:cvc-general-result}

This section is devoted to the proof of our main theorem, Theorem~\ref{thm:cvc-G-deletion-theorem}.
Recall that the decision version of {\cvcGdeletion} is formally defined as follows.

\begin{problem}[]{{\sc CVC parameterized by ${\cG}$-Deletion ({\cvcGdeletion})}}
  \Input & A graph $G$, a vertex set $S$ of size $k$ such that $G - S \in {\cG}$, integer $p$. \\
  \Parameter & \hspace{2 mm} $k$ \\
  \Prob & Does $G$ have a connected vertex cover of size at most $p$?
\end{problem}

Note that in our problem description, we explicitly require a ${\cG}$-deletion set to be given with the input.
The formal definition of the parameterized optimization version of this problem is as follows, where the input is the tuple $(G,S,k)$.

\smallskip
$CVC{\sf \mbox{-}{\cG}\mbox{-}Del}(G, S,k,T) = \begin{cases} -\infty &\mbox{if } |S| > k \mbox{ or } G - S \notin {\cG}\mbox{,}\\
\infty &\mbox{if } T \mbox{ is not a connected vertex cover of }G\mbox{,}\\
|T| &\mbox{otherwise.} \end{cases}$

We use $-\infty$ to denote the malformed input instances and $\infty$ to denote the infeasible solutions.
We need to assume that the recognition problem for the graph class ${\cG}$ is polynomial-time solvable in order to identify any malformed input instance. We let $\OPT (G)$ denote the size of a smallest connected vertex cover of a connected graph $G$.
When $G$ is clear from the context, we simply write $\OPT$.
The following observation is a property of optimal connected vertex covers we will use throughout the paper.

\begin{observation}
\label{obs:contraction-of-cvc}
Let $G$ be a connected graph and $S \subseteq V(G)$.
Let $\hat{G}$ be the graph obtained from $G$ by identifying the vertex set $S$ into ${\hu}$.
Then, $\OPT (\hat{G}) \leq \OPT (G)$.
\end{observation}

\begin{proof}
Let $X^\star$ be an optimal connected vertex cover of the connected graph $G$.
By construction of $\hat{G}$, we know that $N_G(S) = N_{\hat G}(\hu)$.
If $X^\star \cap S = \emptyset$, then we know that $N_G(S) \subseteq X^\star$ and $X^\star$ is still a connected vertex cover of $\hat G$.
On the other hand, if $X^\star \cap S \neq \emptyset$, then we construct $\hat X = (X^\star \setminus S) \cup {\hu}$.
Clearly, all edges of $G$ incident to $S$ are incident to ${\hu}$.
Furthermore, identifying a set of vertices into a single vertex preserves connectivity and $|\hat X| \leq |X^\star|$.
So, $\hat X$ is still a connected vertex cover of $\hat G$.
Hence, $\OPT (\hat G) \leq |\hat X| \leq |X^\star| = \OPT (G)$.
\end{proof}

The above observation guarantees that the identification of a set of vertices does not increase the optimal solution size.
Using Observation~\ref{obs:contraction-of-cvc}, we prove the following lemma that will be crucial for the main theorem of this section, i.e. Theorem~\ref{thm:cvc-G-deletion-theorem}.
The following lemma guarantees that whenever the deletion set $S$ (this means that $G-S\in \cG$) is known to be sufficiently small compared to
a $2$-approximate connected vertex cover of $G$, we can compute a $(1+\eps)$-approximate connected vertex cover of $G$ {\em containing} $S$ in polynomial-time for every fixed $\eps > 0$.
Note that $\alpha = 1 + \eps$ and throughout the paper, we use $\alpha$ and $(1+\eps)$ interchangeably.

\begin{lemma}
\label{lem:easy-algorithm-G-deletion}
Let $0 < \eps < 1$ and $G$ be a connected graph and $S \subseteq V(G)$ such that $G - S \in {\cG}$.
If $|S| \leq \frac{\eps}{6}|L|$ where $L$ is a $2$-approximate connected vertex cover of $G$, then there exists a polynomial-time algorithm ${\cB}$ which takes as input $G$ and $S$ and satisfies the following properties. 

\begin{enumerate}	
	\item\label{case2:ptas}
If {\CVC} admits a {\sc PTAS} on the graph class ${\cG} + 1\cdot {\sf v}$, then Algorithm ${\cB}$ runs in time $f(1/{\eps})n^{g(1/{\eps})}$ for some computable functions $f$ and $g$ and outputs a connected vertex cover of $G$ that contains $S$ and whose size is at most $(1+\eps)~\OPT (G)$.

	\item\label{case1:eptas} 
	If {\CVC} admits an {\sc EPTAS} on the graph class ${\cG} + 1\cdot {\sf v}$, then Algorithm ${\cB}$ runs in time $f(1/{\eps})n^{\OO(1)}$ for some computable function $f$ and outputs a connected vertex cover of $G$ that contains $S$ and whose size is at most $(1+\eps)~\OPT (G)$.
\end{enumerate}
\end{lemma}

\begin{proof}
As $|S| \leq \frac{\eps}{6}|L|$ where $L$ is a $2$-approximate connected vertex cover of $G$, we have that $|S| \leq \frac{\eps}{3}\OPT (G)$.
Let ${\hG}$ be the graph obtained by identifying the vertex set $S$ into a single vertex ${\hu}$.

If {\CVC} admits an {\sc EPTAS} in ${\hG}$ (Condition~\ref{case1:eptas} holds true), then we know that given $0< {\eps} < 1$, there exists an algorithm $\cB$ that runs in time $f(1/{\eps})n^{\OO(1)}$ and outputs a connected vertex cover $X^\star$ of size at most $(1+\frac{\eps}{3})\OPT (\hG)$.

If {\CVC} admits a {\sc PTAS} in ${\hG}$  (Condition~\ref{case2:ptas} holds true), then we know that given $0 < \eps < 1$, there exists an algorithm $\cB$ that runs in time $f(1/{\eps})n^{g(1/{\eps})}$ and outputs a connected vertex cover $X^\star$ of size at most $(1+\frac{\eps}{3}) \OPT (\hG)$.

We know by Observation~\ref{obs:contraction-of-cvc} that $\OPT (\hG) \leq \OPT (G)$.
Let us consider the set $X = (X^\star \setminus \{{\hu}\}) \cup S$.
As $N_{\hat{G}}({\hu}) = N_G(S)$, we know that the set of edges of $G$ incident on $S$ are covered by $X$.
So, irrespective of whether ${\hu} \in X^\star$ or ${\hu} \notin X^\star$, we have the fact that $X$ is a vertex cover of $G$.
But $G[X]$ need not be a connected subgraph of $G$.
First, we have to ensure that there are at most $|S| + 1$ connected components in $G[X]$.
There are two cases.

\begin{enumerate}
	\item First, we consider the case when ${\hu} \in X^\star$ and $|X^\star| \geq 2$, we know that ${\hu}$ has at least one neighbor in $X^\star$.
So, at least one vertex of $S$ has at least one neighbor in $X^\star \setminus \{{\hu}\}$.
Otherwise, $X^\star = \{\hu\}$ and in such situation, we have $X = S$.
Hence, irrespective of $|X| = 1$ or $|X| \geq 2$, the graph $G[X]$ has at most $|S|$ connected components.

	\item Now, we consider the case when ${\hu} \notin X^\star$.
Then, we know that $X = (X^\star \setminus \{\hu\}) \cup S = X^\star \cup S$.
So, $G[X]$ has at most $|S| + 1$ connected components as $G[X^\star]$ is already connected.
\end{enumerate}

Now, we have that $G[X]$ has at most $|S| + 1$ connected components.
In order to convert $X$ into a connected subgraph of $G$, we now add some additional vertices from $G - X$.
We know that $G - X$ is an independent set and $G$ is connected.
If $G[X]$ has more than one connected component, there exists a vertex $v \in V(G - X)$ such that $v$ is adjacent to at least two different connected components of $G - X$.
We find such a vertex and add it to $X$.
We continue this process until we have that $X$ induces a connected subgraph.
As there are at most $|S| + 1$ connected components in $G[X]$, we will need to repeat this step at most $|S|$ times.
So, the size of the final connected vertex cover we generate is at most $|X| + |S| \leq |X^\star| + 2|S| \leq (1 + \frac{\eps}{3})\OPT (\hat{G}) + 2|S| \leq (1 + \frac{\eps}{3})\OPT (G) + \frac{2\eps}{3} \OPT (G) = (1+\eps)\OPT (G)$.

Now note that if Condition~\ref{case1:eptas} holds true, then the running time of the algorithm to compute $X^\star$ is $f(1/{\eps})n^{\OO(1)} + n^{\OO(1)}$.
On the other hand if Condition~\ref{case2:ptas} holds true, then the running time of the algorithm to compute $X^\star$ is $f(1/{\eps})n^{g(1/{\eps})} + n^{\OO(1)}$.
All subsequent steps of our algorithm run in polynomial-time.
As a result, this proves our claimed bounds on the running times and completes the proof of the lemma.
\end{proof}

Now, we consider the case when $|S| > (\eps/6) |L|$, where $L$ is a $2$-approximate connected vertex cover of $G$.
In this case, we have that $|S| > (\eps/6)\OPT(G)$.
So, we know that $\OPT(G) < \ceil{6k/\eps}$.
In that case, we can modify the PSAKS provided by Lokshtanov et al.~\cite{LokshtanovPRS17}, but with parameter value $\ceil{6k/\eps}$.
We give a proof of this in the following lemma.

\begin{lemma}
\label{lemma:high-G-deletion-parameter}
Let $(G, S, k)$ be the given instance of {\cvcGdeletion}.
Given a $2$-approximate connected vertex cover of $G$, say $L$, if $|S| > (\eps/6)|L|$, then one can compute a graph $G'$ in polynomial-time such that the following statements hold. 
\begin{enumerate} 
    \item $|V(G')|=k^{\OO(1/\eps)}$,
    \item $G'$ has a connected vertex cover of size at most $(1+\eps)\OPT $.
    \item Every inclusion-wise minimal connected vertex cover of $G'$ is a connected vertex cover of $G$.
\end{enumerate}
\end{lemma}

\begin{proof}
If $|S| > \frac{\eps}{6}|L|$, then it follows that $|S| > \frac{\eps}{6}\OPT$.
That is, $\OPT < \ceil{6k/\eps}$.
The graph $G'$ is the output of a slightly modified version of the PSAKS for {\CVC} parameterized by solution size~\cite{LokshtanovPRS17} but with the slightly different parameter value $k' = \ceil{6k/\eps}$.
Since, $k'$ is linearly bounded in $k$, we will be able to conclude that $|V(G')|$ is $k^{\OO(1/{\eps})}$.
We present a brief sketch of the construction of $G'$ for the sake of completeness.

Let $H$ denote the set of vertices of $G$ whose degree is at least $k'+1$. Observe that every vertex cover of $G$ of size at most $k'$ must contain every vertex in $H$. 
Moreover, since $G$ has a vertex cover of size $\OPT$ and $\OPT$ is upper bounded by $k'$, we conclude that $|H|\leq k'$.

Let $R$ denote the set of vertices in $V(G)\setminus H $ which have at least one neighbor which is not in $H$. Since all such vertices have degree at most $k'$ and $G$ has a vertex cover of size at most $k'$, we conclude that $|R|\leq (k')^2$. 

Let $I=V(G)\setminus (H\cup R)$. By definition, $G[I]$ is an independent set. 
The PSAKS of Lokshtanov et al.~\cite{LokshtanovPRS17} uses as a subroutine an efficient algorithm that computes a set $I'\subseteq I$ of size $(k')^{\bigoh(1/\eps)}$ such that $G'=G[I'\cup H\cup R]$ satisfies the properties required by the lemma.
This subroutine essentially does the following. First of all, for every $h\in H$, if $h$ has at most $k'+1$ neighbors in $I$, then it marks all of these neighbors. Otherwise, it marks an arbitrarily chosen set of $k'+1$ neighbors of $h$ in $I$. This is done simply to preserve the status of vertices in $H$ as vertices of degree at least $k'+1$. It then repeatedly executes the following step as long as possible. If there is an unmarked vertex $v\in I$ such that $v$ is neighbor to at least $1/\eps$ distinct {\em connected components} of $G[H]$, then mark this vertex, contract all edges incident on $v$, and add the resulting new vertex to $H$. Observe that this will reduce the number of connected components of $G[H]$ by at least $1/\eps-1$ and so this step will not be repeated more than $\eps\cdot |H|$ times. When this procedure terminates, we go over the remaining unmarked vertices in $I$ and for every set of vertices with the same neighborhood in $H$, mark one of these vertices and remove all others. Since any surviving vertex has degree at most $1/\eps$ into the (modified) set $H$, the number of vertices marked in this procedure is $(k')^{\bigoh(1/\eps)}$ these vertices form the set $I'$. Finally, we add pendants to every vertex in $H$ to force them into every connected vertex cover of $G'$. 
It follows from the definition of $I'$ that a connected vertex cover $Z$ of $G$ of size $\OPT$ can be converted to one of size at most $(1+\eps)|Z|$ for $G'$ and consequently, every inclusion-wise minimal
 connected vertex cover of $G'$ (which must be disjoint from the pendant vertices added in the end) is also a connected vertex cover of $G$.
This completes the proof of the lemma.
\end{proof}

We are now ready to combine Lemma~\ref{lem:easy-algorithm-G-deletion} and Lemma~\ref{lemma:high-G-deletion-parameter} to obtain Theorem~\ref{thm:cvc-G-deletion-theorem}.

\CVCGeneralTheorem*

\begin{proof}
Let $(G, S, k)$ be the given instance of {\cvcGdeletion}, where $G - S \in {\cG}$.
We first state the approximate kernelization algorithm.
Recall that the approximate kernelization algorithm must have two parts.
The first part is the {\em Reduction Algorithm} and the second is the {\em Solution Lifting Algorithm}.
\begin{itemize}
  \item {\bf Reduction Algorithm:}
 We use the algorithm by Savage~\cite{Savage82} to compute a 2-approximate connected vertex cover of $G$, call it $L$. If 
$|S|\leq \frac{\eps}{6}\cdot |L|$, then we return a trivial instance (of constant size) $(G',S,k')$ of {\cvcGdeletion} and otherwise, we invoke Lemma~\ref{lemma:high-G-deletion-parameter} to compute the subgraph $G'$ and return the instance $(G',S,k')$, where $k'=\lceil 6k/\eps\rceil$. Note that if $S$ is not completely contained in $G'$, then we may simply add it back to $G'$. Since $G'$ is a subgraph of $G$ and $\cG$ is hereditary, we know that $S$ is also a $\cG$-deletion set of $G'$.
Clearly, the {\em size} of the output satisfies the required bound. It only remains to prove the correctness of the reduction by providing a solution lifting algorithm. 

\medskip

    \item {\bf Solution Lifting Algorithm:}
    Recall that the solution lifting algorithm has access to $G$. We may also assume without loss of generality that it has access to the set $L$, which was computed by the reduction algorithm. Let $Q$ be the given $c$-approximate solution for the instance output by the reduction algorithm. We may assume without loss of generality that $Q$ is inclusion-wise minimal. 

If $|S|\leq \frac{\eps}{6}\cdot |L|$, then we ignore the set $Q$ and invoke Lemma \ref{lem:easy-algorithm-G-deletion} to compute and return a $(1+\eps)$-approximate connected vertex cover of $G$. On the other hand, if $|S|> \frac{\eps}{6}\cdot |L|$ we simply return $Q$ and use Lemma \ref{lemma:high-G-deletion-parameter}~(2) and Lemma \ref{lemma:high-G-deletion-parameter}~(3) to conclude that $|Q| \leq c (1+\eps)\cdot \OPT$.
\end{itemize}

There are two cases, one where {\CVC} admits a {\sc PTAS} and the other where {\CVC} admits an {\sc EPTAS} on every graph $G \in {\GG} + 1\cdot {\sf v}$.
In each case, we have managed to convert a $c$-approximate solution of $G'$ to a $c(1+\eps)$-approximate solution of $G$.
Now, we prove the items in the given order.
\begin{enumerate}
	\item Suppose that {\CVC} admits a {\sc PTAS} on $G \in {\GG} + 1\cdot {\sf v}$ (Condition~\ref{case2:ptas} holds).
	If $|S| \leq \frac{\eps}{6}|L|$, then by Lemma~\ref{lem:easy-algorithm-G-deletion}, we can find a connected vertex cover of size at most $(1+\eps)\OPT (G)$ in $f(1/{\eps})n^{g(1/{\eps})}$ time.
	So, based on that we have a PSAKS running in $f(1/{\eps})n^{g(1/{\eps})}$ time.
	
	\item On the other hand, suppose that {\CVC} admits an {\sc EPTAS} on $G \in {\GG} + 1\cdot {\sf v}$. If $|S| \leq \frac{\eps}{6}|L|$, we know by Lemma~\ref{lem:easy-algorithm-G-deletion} that there exists an algorithm that runs in time $f(1/{\eps})n^{\OO(1)}$ that computes a $(1+\eps)$-approximate connected vertex cover of $G$.
	So, in this case, the reduction algorithms and solution lifting algorithms run in $f(1/{\eps})n^{\OO(1)}$ time.
	Hence, we have a time efficient PSAKS.
\end{enumerate}
This completes the proof of the theorem.
\end{proof}

\subsection{{\CVC} parameterized by $\eta$-Transversal Number}
\label{sec:cvc-treewidth}
Now, we consider the problem {\cvctreewidth} whose decision version is as follows.

\begin{problem}[]{{\sc CVC parameterized by $\eta$-transversal~({\cvctreewidth})}}
  \Input & A graph $G$, an $\eta$-transversal $S$ of size $k$, integer $p$. \\
  \Parameter & \hspace{2 mm} $k$ \\
  \Prob & Does $G$ have a connected vertex cover of size at most $p$?
\end{problem}

The optimization version of {\cvctreewidth} is as follows.

$CVC{\sf \mbox{-}{\eta}\mbox{-}tvl}(G, S,k,T) = \begin{cases} -\infty &\mbox{if } |S| > k \mbox{ or } \tw (G - S) > \eta\mbox{,}\\
\infty &\mbox{if } T \mbox{ is not a connected vertex cover of }G\mbox{,}\\
|T| &\mbox{otherwise.} \end{cases}$

We know that {\CVC} is polynomial-time solvable on graphs of bounded treewidth.
Let $(G, S, k)$ be a given instance where $|S| \leq k$ and $\tw (G - S) \leq \eta$ for some constant $\eta$.
So, we have Corollary~\ref{cor:cvc-treewidth-result} as a consequence of Theorem~\ref{thm:cvc-G-deletion-theorem}.

\TreewidthResult*

\begin{proof}
It is known due to Flum and Grohe~\cite{FG06pcbook} that there exists a polynomial-time algorithm that can check if a graph has constant treewidth.
From the assumption, $\tw (G - S) \leq \eta$.
Let $F = G - S$ and consider the graph $F + 1\cdot {\sf v}$.
We know that $\tw (F + 1\cdot {\sf v}) \leq \eta + 1$.
We know from Courcelle's Theorem~\cite{Courcelle90} that there exists a linear time algorithm that outputs a minimum connected vertex cover of $F + 1\cdot {\sf v}$.
Now, using Theorem~\ref{thm:cvc-G-deletion-theorem}, we see that {\cvctreewidth} admits a time efficient PSAKS with $k^{\OO(1/{\eps})}$ vertices.
\end{proof}

\subsection{{\CVC} parameterized by Chordal Deletion Number}
\label{sec:cvc-chordal-deletion}
Now, we consider the following problem {\cvcchordal} whose decision version is stated as follows.

\begin{problem}[]{{\sc CVC parameterized by chordal deletion number~({\cvcchordal})}}
  \Input & A graph $G$, a chordal vertex deletion set $S$ of size $k$, integer $p$. \\
  \Parameter & \hspace{2mm} $k$ \\
  \Prob & Does $G$ have a connected vertex cover of size at most $p$?
\end{problem}

Note that as in the case of {\cvcGdeletion}, we explicitly require a chordal vertex deletion set to be provided as part of the input. However, this requirement can be removed and replaced with an execution of the polynomial-time factor-$\OPT^{\bigoh(1)}$ approximation algorithm for {\sc Chordal Vertex Deletion} of Jansen and Pilipczuk~\cite{JP17}. 

The formal definition of the parameterized optimization version of this problem is as follows:

\smallskip
$CVC\textsf{-chordal}(G,S,k,T) = \begin{cases} -\infty &\mbox{if } |S| > k \mbox{ or some component of } G-S \notin \textsf{Chordal}\mbox{,}\\
\infty &\mbox{if } T \mbox{ is not a connected vertex cover of }G\mbox{,}\\
|T| &\mbox{otherwise.} \end{cases}$

\smallskip

Let $\eta=2+\lceil \frac{1}{\eps}\rceil$.
We apply the following reduction rule.

\begin{reduction rule}
\label{red:compress-maximal-clique}
Let $(G,S,k)$ be the given instance of {\cvcchordal}, where $G-S$ is a chordal graph. 
Let $C\subseteq V(G\setminus S)$ such that $|C| \geq \eta$ and $G[C]$ is a maximal clique.
Contract the edges of $G[C]$ to obtain a new vertex $u_C$ and add a pendant vertex $v_C$ adjacent to $u_C$.
Let $G'$ denote the graph resulting from this operation.
\end{reduction rule}

The intuition behind this reduction rule comes from the fact that since any vertex cover must contain all but at most one vertex of this clique, we can also force the remaining vertex of the clique (if one exists) into the solution at the cost of a small but manageable error without violating the connectivity requirement.

\begin{lemma}
\label{lem:safeness-compress-maximal-clique-rule}
Let $G'$ be the resulting graph after applying Reduction Rule~\ref{red:compress-maximal-clique} on $G$.
Then, there is a polynomial-time algorithm that, given a $c$-approximate connected vertex cover of $G'$, returns a connected vertex cover of $G$ whose size is at most $\beta\cdot \OPT(G)$ where $\beta= \max\{c,(1+\eps)\}$.
\end{lemma}

\begin{proof}
Let $D'$ be the given connected vertex cover of $G'$ and let $\OPT'$ denote the size of a smallest connected vertex cover of $G'$. Recall that $\OPT=\OPT(G)$ denotes the size of a smallest connected vertex cover of $G$.
Note that $u_C \in D'$ because it has a pendant neighbor $v_C$. In addition, since $u_C$ is the unique neighbor of $v_C$, we may assume without loss of generality that $v_C\notin D'$. 
We now argue that $D = (D' \setminus \{u_C\}) \cup C$ is the required connected vertex cover of $G$.
It is straightforward to see that $D$ is a connected vertex cover of $G$. 

It remains to prove that $|D|\leq \beta \cdot \OPT$. 
Since any vertex cover of $G$ must contain at least $|C|-1$ vertices of the clique $C$, 
we infer that $\OPT \geq \OPT' + |C| - 2=\OPT'+\eta-2$.
And by definition, $|D| \leq |D'| + \eta - 1$. Combining these two inequalities, we obtain the following. 

$$\frac{|D|}{\OPT} \leq \frac{|D'| + \eta - 1}{\OPT' + \eta - 2} \leq \textsf{max}\Bigg\{\frac{|D'|}{\OPT'}, \frac{\eta-1}{\eta-2}\Bigg\} \leq \textsf{max}\Big\{c,(1+\eps)\Big\}=\beta$$

Hence we conclude that $|D|\leq \beta\cdot \OPT$. This completes the proof of the lemma. 
\end{proof}

When Reduction Rule~\ref{red:compress-maximal-clique} is not applicable, it must be the case that all the maximal cliques of $G - S$ are of size at most $\eta = 2 + \ceil{1/\eps}$. This gives us the following lemma.

\begin{lemma}
\label{lem:no_large_cliques}
(See~\cite{combinatorial-optimization,Golumbic})
	If $G$ has a set $Z$ such that $G-Z$ is chordal and $G-Z$ has no cliques of size $\eta$, then $Z$ is also an $(\eta-2)$-transversal for $G$.
\end{lemma}

\begin{proof}
	The lemma follows from the fact that the treewidth of the chordal graph $G-Z$ is upper bounded by the size of the maximum clique in $G-Z$ (see~\cite{Golumbic}), which is at most $\eta-2$.
	As a result, $Z$ is also an $(\eta - 2)$-transversal of $G$.
\end{proof}

Now, combining Corollary~\ref{cor:cvc-treewidth-result}, Lemma~\ref{lem:safeness-compress-maximal-clique-rule} and Lemma~\ref{lem:no_large_cliques}, we are ready to prove Theorem~\ref{thm:chordal_main}.

\chordalmaintheorem*

\begin{proof}

Let $(G,S,k)$ be the given instance of the optimization version of {\cvcchordal}. Recall that $\eta=2+\lceil \frac{1}{\eps}\rceil$. We apply Reduction Rule~\ref{red:compress-maximal-clique} exhaustively on $G-S$ to obtain a graph $G'$ such that $S$ is a chordal vertex deletion set of $G'$ and $G'-S$ has no cliques of size $\eta$. Moreover, it follows from the description of the rule that if $G$ is connected, so is $G'$.
Observe that Reduction Rule~\ref{red:compress-maximal-clique} is approximation preserving.
Hence, by Lemma \ref{lem:safeness-compress-maximal-clique-rule}, it follows that any given $c$-approximate connected vertex cover of $G'$ can be converted to a connected vertex cover of $G$ whose size is at most $\beta\cdot \OPT(G)$ in polynomial-time, where $\beta =\max\{c,(1+\eps)\}$.
In addition, from Lemma \ref{lem:no_large_cliques}, we know that $\tw(G-S)\leq \eta$, implying that we may now treat $(G',S,k)$ as a meaningful instance of {\cvctreewidth}.
Also from Theorem~\ref{lemma:chordal-clique-tree}, the number of maximal cliques in a chordal graph is polynomially upper bounded.
Furthermore, the clique tree decomposition can also be found in polynomial-time.
Hence, all these procedures now run in polynomial-time. 

We now invoke Corollary~\ref{cor:cvc-treewidth-result} and return as our output, the output of the associated PSAKS when given the instance $(G',S,k)$ as input. The correctness as well as the upper bound on the size of the returned output follow from those of Corollary~\ref{cor:cvc-treewidth-result}.
We note that even if $k$ is the size of a smallest chordal vertex deletion set of the original input graph $G$ and $S$ is only a factor-$\OPT^{\bigoh(1)}$ approximation, it follows that the size of the output is still bounded polynomially in $k$, since $|S|$ would be upper bounded polynomially in $k$. This completes the proof of the theorem.
\end{proof}

\section{{\CVC} parameterized by the deletion distance to split graphs or cographs}\label{sec:diameter}

In this section we present a PSAKS for {\CVC} parameterized by the size of a minimum deletion set into a disjoint union of split graphs and cographs.

\begin{lemma}
\label{Lem:characterization-of-Fd}
Let $G$ be a graph such that there exists a vertex $u$ such that every connected component of $G - \{u\}$ is either a split graph or a cograph.
Then, {\CVC} is polynomial-time solvable on $G$.
Furthermore, there exists a polynomial-time algorithm that outputs a smallest connected vertex cover containing $u$.
\end{lemma}
	
	\begin{proof} 
It is sufficient for us to provide a polynomial-time algorithm to solve inputs of {\CVC} on the class $\CC\CC(\cG)+1\cdot {\sf v}$ such that $\cG = \textsf{ Split} \cup  \textsf{Cograph}$.
	 Recall that the inputs $(G,u^\star)$ we are interested in, have the property that $G-u^\star\in \CC\CC(\cG)$ and we are only interested in smallest connected vertex covers of $G$. 
	
\begin{description}
		\item[Case 1:] 
		$\cG=\textsf{Split}$. Consider a graph $G$ in the class $\cG+1\cdot {\sf v}$ and let $u^\star$ denote a vertex in $G$ such that $G-u^\star$ is a split graph.
		The proof idea is similar to the way presented by Krithika et al.~\cite{KMR18}, but we provide it here for completeness.
		 Let $(G,u^\star)$ be the given instance of {\CVC} and let $(C,I)$ denote the partition of $V(G)\setminus \{u^\star\}$ into a clique and independent set respectively.
We will now construct a connected vertex cover $Z'$ by making a constant number of non-deterministic choices as follows.
Initialize $Z':=\{u^\star\}$ when we include $u^\star$ in the vertex cover, and otherwise we initialize $Z' = \emptyset$ (when we consider not to include $u^\star$ in the vertex cover).
Since any vertex cover must pick at least $|C|-1$ vertices from $C$, we guess whether or not $Z$ contains all of $C$ and in the latter case, we guess the unique vertex $q\in C\setminus Z$. 
In case, we decide not to pick $u^\star$ in the vertex cover, we pick its neighbors present in $C$ and $I$.
If the current set $Z'$ is a connected vertex cover of $G$, then we simply return it.
Otherwise, it must be the case that it has two different components.
They can be connected by using the only possible remaining vertex in $C$ that we might not have picked. So, we add that vertex to $Z'$ and check if that forms a connected vertex cover.
If that forms a connected vertex cover, then we return $Z = Z'$, or otherwise we simply return $\bot$ to indicate that the graph has no connected vertex cover. 

When $G - \{u^\star\}$ has more than one connected component, then $u^\star$ must always be there in any connected vertex cover.
In such case, let $C$ be a connected component of $G - \{u^\star\}$.
Consider any minimum connected vertex cover $A$ of $G$.
Consider $A_C = A \cap (V(C) \cup \{u^\star\})$.
$A_C$ must also be a vertex cover of $G[V(C) \cup \{u^\star\}]$.
Furthermore, $A_C$ also must be connected as well since $u^\star$ is a cut vertex in $G[A]$ also.
In such case, we compute a minimum connected vertex cover $A_C$ of each $G[V(C) \cup \{u^\star\}]$, and finally we put together to get a minimum connected vertex cover of $G$.

\smallskip

\item[Case 2:] $\cG=\textsf{Cograph}$. 
For the case of cographs, note that they are a subclass of distance-hereditary graphs and hence have rankwidth 1~\cite{Oum05}.
Therefore, {\CVC} is polynomial-time solvable on the class $\CC\CC(\cG)+1\cdot {\sf v}$ when $\cG=\textsf{Cograph}$, is a direct consequence of the fact that this problem is expressible in $\textsf{MSO}_1$, the result of Courcelle, Makowsky, and Rotics~\cite{CK09}, and the fact that adding a constant number of vertices to a graph of constant rankwidth keeps the rankwidth of the resulting graph constant~\cite{Oum05,CK09}.
In fact, by the same result of Courcelle, Makowsky, and Rotics~\cite{CK09}, we can find a smallest connected vertex cover containing $u^\star$ in polynomial-time. 
We note that this argument cannot be used in the previous case since split graphs can have arbitrarily large rankwidth~\cite{Oum05,CK09}. 
\end{description}

Suppose that $G - \{u^\star\}$ has one connected component being a split graph and another connected component being a cograph.
Since $G$ is connected, without loss of generality, we can assume that there are at least two connected components $C_1$ and $C_2$ of $G - \{u^\star\}$ such that there is an edge in $G[\{u^\star\} \cup V(C_1)]$ as well as in $G[\{u^\star\} \cup V(C_2)]$.
Then any (smallest) connected vertex cover of $G$ must contain $u^\star$.
In that case, for each connected component $C$, we first compute a minimum connected vertex cover of $G[\{u^\star\} \cup V(C)]$ containing $u^\star$.
Then, we compute their union to get a minimum connected vertex cover of $G$.
This completes the proof of the lemma.
\end{proof}

\begin{problem}[]{{\sc CVC-Split-Cograph-Deletion}}
  \Input & A graph $G$, a set $S \subseteq V(G)$ of size $k$ such that every connected component of $G - S$ is either a split graph or a cograph and an integer $p$. \\
  \Parameter & \hspace{2mm} $k$ \\
  \Prob & Does $G$ have a connected vertex cover of size at most $p$?
\end{problem}

For this problem, we assume the deletion set is provided in the input.
The formal definition of the parameterized optimization version of this problem is as follows:

\smallskip
$CVC(G, S,k,T) = \begin{cases} -\infty &\mbox{if } |S| > k \mbox{ or } G - S \notin \textsf{Split} \cup \textsf{Cograph}\mbox{,}\\
\infty &\mbox{if } T \mbox{ is not a connected vertex cover of }G\mbox{,}\\
|T| &\mbox{otherwise.} \end{cases}$

\smallskip

Let $(G, S, k)$ be an instance of {\sc CVC-Split-Cograph-Deletion}.
Let $\hat G$ be the graph constructed by identifying the vertex set $S$ into a single vertex $\hat s$.
We know from Lemma~\ref{Lem:characterization-of-Fd} that {\CVC} is polynomial-time solvable on $\hat G$.
Now, using Theorem~\ref{thm:cvc-G-deletion-theorem} and Lemma~\ref{Lem:characterization-of-Fd}, we prove Corollary~\ref{thm:diameter_main}.

\diametermaintheorem*

We have designed our kernelization algorithms in Sections~\ref{sec:cvc-general-result} and~\ref{sec:diameter} in such a way, that our lossy kernels can be unified under a single parameterization.
Hence, we consider a single parameterization, the {\em deletion distance} of the input graph to the class of graph where every connected component is {\em either} a treewidth-$\eta$ graph, {\em or} a chordal graph, {\em or} a cograph.
So, essentially we consider {\CVC} on $\CC\CC(\cG) + k\cdot {\sf v}$ when $\cG = \textsf{Treewidth-$\eta$-Graph} \cup \textsf{Chordal-Graph} \cup \textsf{Cographs}$.
We assume that the deletion set $S$ is given with the input.
This assumption is important as no $\OPT^{\OO(1)}$ polynomial-time approximation algorithm is known to find such a deletion set.

We know that the rankwidth and the cliquewidth of a graph are equivalent.
It means that, if a graph has bounded cliquewidth, then it also has bounded rankwidth and vice versa.
Therefore, we first prove the following lemma, which we will use in the proof of Theorem~\ref{thm:more-general-theorem-statement}.

\begin{lemma}
\label{lem:cvc-easy-solvability-mixed-parameter}
Let $\eta$ be a constant and $G$ be a connected graph having a vertex $u$ such that every connected component of $G - \{u\}$ is either a chordal graph, or a graph with cliquewidth $\eta$.
Then, {\CVC} is polynomial-time solvable on $G$.
\end{lemma}

\begin{proof}
First, we partition the set of connected components of $G - \{u\}$ as follows.
Let
\begin{itemize}    
    \item $\cD_1$ be those connected components of $G - u$ that are chordal graphs, and
    
    \item $\cD_2$ be those connected components of $G - u$ that have cliquewidth at most $\eta$ but are not in ${\cD}_1$.
\end{itemize}

For every $i \in \{1, 2\}$, we denote by $G_i$ the graph induced by the vertex set spanned by the connected components in $\cD_i$ plus the vertex $u$.
As $G$ is connected, without loss of generality, we can assume that both the graphs $G_1$ and $G_2$ have an edge.
So, $u$ is part of any (optimal) connected vertex cover of $G$.
For every $i = \{1, 2\}$, let $B_i$ be a smallest connected vertex cover of $G_i$ that contains $u$.
We define $\cH$ by the class of all connected graphs that contain a vertex whose removal results in a chordal graph.
Note that $G_1 \in \cH$.
Suppose that a new vertex $w$, and an edge $wv$ such that $v \in V(G_1)$ are added to $G_1$ (called {\em pendant addition} by Escoffier et al.~\cite{EscoffierGM10}).
Even then also $G_1 - \{u\}$ is a chordal graph.
So, $\cH$ is closed under pendant addition.
So, each of the biconnected components of $G_1$ also has one vertex whose deletion results in a chordal graph. 
Hence, all the biconnected components of $G_1$ are also in the graph class $\cH$.
It is known due to Lemma 4 of Escoffier et al.~\cite{EscoffierGM10} that finding a minimum a connected vertex cover and finding smallest connected vertex cover containing $u$ are polynomially equivalent in $G_1$.
Now, we explain how {\CVC} is polynomial-time solvable on $G_1$.
The idea is standard and goes along the line of the proof of Theorem 23 by Krithika et al.~\cite{KrithikaMR17}.
It is known due to Theorem~\ref{lemma:chordal-clique-tree} that every chordal graph admits a clique tree decomposition that can be obtained in polynomial-time.
Now, adding $u$ to every bag gives a valid tree decomposition of $G_1$.
Now, we convert this into a nice tree decomposition of $G_1$.
Let $\cT = (T, \{X_t\}_{t \in V(T)})$ be a nice tree decomposition of $G_1$.
Every bag of this nice tree decomposition $\cT$ has at most one vertex whose deletion results in a clique.
Consider a standard dynamic programming over a rooted nice tree decomposition $\cT$ of $G_1$.
For a node $t$ of $T$, let $T_t$ be the subtree rooted at $t$ and $V_t$ denotes the union of all bags rooted at $t$.
Furthermore, let $H_t$ be the subgraph of $G_1$ induced by the vertex set $V_t$.
For a node $t$ of $T$, and $X \subseteq X_t$, and a partition $\mathcal{P} = \{P_1,\ldots,P_q\}$ of $X$ into at most $|X|$ parts, let $VC[t, X, \mathcal{P}]$ denote a minimum vertex cover $Y$ of $H_t$ such that $Y \cap X_t = X$, and $H_t[Y]$ has exactly $q$ connected components $C_1,\ldots, C_q$ where for all $i \in [q], P_i = V(C_i) \cap X_t$.
Moreover, if $X = \emptyset$, then $Y$ is required to be connected in $H_t$.
So, consider the root node $r$ of a nice tree decomposition.
We know that $X_r = \emptyset$.
Now, $VC[r, \emptyset, \{\emptyset\}]$ is a minimum connected vertex cover of $G_1$.
But, every bag has one vertex whose deletion results in a clique.
So, any connected vertex cover of $G_1$ can avoid at most two vertices from $X_t$ where $t \in V(T)$ (from every bag of $\cT$).
Hence, the total number of valid states per node is $|V(G_1)|\cdot |X_t|^{\OO(1)}$.
So, each entry can be computed in polynomial-time.
Hence, {\CVC} is polynomial-time solvable on $G_1$.
As {\CVC} and finding a smallest connected vertex cover containing $u$ are polynomially equivalent in $G_1$, we can find $B_1$ in polynomial-time.

We know that $\cw (G_2) \leq \eta + 1$, we know due to Courcelle's Theorem~\cite{Courcelle1990,Courcelle90} that {\CVC} is polynomial-time solvable on $G_2$.
In fact, due to the same result by Courcelle~\cite{Courcelle1990,Courcelle90}, a smallest connected vertex cover containing $u$ can also be computed in polynomial-time in $G_2$. 
So, $B_2$ can be found in polynomial-time.

Hence, $B_1 \cup B_2$ is an optimal connected vertex cover of $G$.
In this process, we can solve {\CVC} in polynomial-time in $G$.
\end{proof}

Note that a graph with treewidth $\eta$ has cliquewidth at most $3\cdot 2^{\eta - 1}$~\cite{CornelRotics2006}.
Since cographs have bounded rankwidth, it has bounded cliquewidth too.
Lemma~\ref{lem:cvc-easy-solvability-mixed-parameter} gives a guarantee that {\CVC} is polynomial-time solvable on any graph in the graph class $\CC\CC(\cG) + 1\cdot {\sf v}$ with $\cG = \textsf{ \tw-$\eta$-Graph } \cup \textsf{Chordal-Graph } \cup \textsf{Cograph}$. 
Now, using the above lemma (Lemma~\ref{lem:cvc-easy-solvability-mixed-parameter}) and Theorem~\ref{thm:cvc-G-deletion-theorem}, we give a proof of Theorem~\ref{thm:more-general-theorem-statement}.

\FinalResultMain*

\begin{proof}
Let $\hat G$ be the graph constructed by identifying the vertex set $S$ into a single vertex ${\hu}$.
Let $C$ be an arbitrary connected component of $G - S$ that has treewidth at most $\eta$.
Due to Cornell and Rotics~\cite{CornelRotics2006}, we know that if $\tw (C) \leq \eta$, then $\cw (C) \leq 3\cdot 2^{\eta - 1}$.
So, we have that every connected component of $G - S$ is either a cograph or a graph with cliquewidth at most $3\cdot 2^{\eta - 1}$ or a chordal graph.
Since cographs also have constant cliquewidth, we know that every connected component of $G - S$ is either a graph of bounded cliquewidth or a chordal graph.
We know from Lemma~\ref{lem:cvc-easy-solvability-mixed-parameter} that there exists a polynomial-time algorithm that constructs an optimal connected vertex cover of $\hat G$.
Now, using Theorem~\ref{thm:cvc-G-deletion-theorem}, we know that {\CVC} parameterized by $k$ admits a time efficient {\sc PSAKS} with $k^{\OO(1/{\eps})}$ vertices.
This completes the proof of our final result.
\end{proof}

\section{Conclusion}
\label{sec:conclusion}

In this paper we obtained the first polynomial size approximate kernelization schemes for the {\CVC} problem when parameterized by the deletion distance of the input graph to the class of cographs, the class of bounded treewidth graphs, and the class of chordal graphs. Moreover, they are in fact {\em time efficient} PSAKSes and this raises the natural question of whether one can obtain a {\em size efficient} PSAKS for {\sc Connected Vertex Cover} even when parameterized by solution size. The output of a size efficient PSAKS is required to be bounded by $f(\eps)~k^{\bigoh(1)}$ instead of $k^{f(\eps)}$. 

We designed our kernelizations in such a way as to ensure that our results have been unified under a single parameterization, the deletion distance of the input graph to the class of graphs where every connected component is \emph{either} a treewidth-$\eta$ graph {\em or} a chordal graph {\em or} a cograph. 

This has allowed our framework to capture graph classes which are significantly more general than classes which have a small deletion distance to {\em only one} of $\{\tw\mbox{-}\eta,{\sf Chordal},{\sf Cographs}\}$.
For instance, consider the graph $H$ obtained by taking the disjoint union of $n/2$ cycles of length $5$ each and $n/2$ cliques of size $4$ each. Observe that the deletion distance of $G$ to any one of $\{\tw\mbox{-}2,{\sf Chordal}\}$ is at least $n/2$. On the other hand, the deletion distance of $H$ to the class of graphs where every connected component is \emph{either} a treewidth-$2$ graph {\em or} a chordal graph is 0.
Our framework thus allows one to obtain a $(1+\varepsilon)$-approximate kernel of {\em constant size} for {\CVC} on $H$ since $k=0$. 

As a final remark, we point out that in order to generalize our results in this way for parameterization by deletion distance to $\{\tw\mbox{-}\eta,{\sf Chordal},{\sf Cograph}\}$ {\em even} in the absence of the deletion set in the input, one must first design a polynomial-time factor-$\OPT^{\bigoh(1)}$ approximation algorithm to compute such a deletion set. We leave this as an interesting problem for future research. Such an algorithm would have interesting implications in the study of graph modification problems.



\end{document}